\documentclass[12pt]{article}
\usepackage{amsmath}
\usepackage{times}
\usepackage{graphicx}
\usepackage{color}
\usepackage{multirow}
\usepackage[authoryear]{natbib}
\usepackage{rotating}
\usepackage{bbm}
\usepackage{bm}
\usepackage{latexsym}
\usepackage{amsthm}
\usepackage{amssymb}
\usepackage{amsfonts}
\usepackage{mathtools}
\usepackage{tabularx}
\usepackage{booktabs}

\theoremstyle{definition}
\newtheorem{assumption}{Assumption}[section]

\newtheorem{theorem}{Theorem}

\newtheorem{corollary}{Corollary}

\DeclareMathOperator*{\argmin}{arg\,min}

\newcommand{\captionfonts}{\normalsize}

\makeatletter  
\long\def\@makecaption#1#2{%
  \vskip\abovecaptionskip
  \sbox\@tempboxa{{\captionfonts #1: #2}}%
  \ifdim \wd\@tempboxa >\hsize
    {\captionfonts #1: #2\par}
  \else
    \hbox to\hsize{\hfil\box\@tempboxa\hfil}%
  \fi
  \vskip\belowcaptionskip}
\makeatother   

\begin{document}
\hspace{13.9cm}1

\ \vspace{20mm}\\

{\LARGE Heteroscedastic Double Bayesian Elastic Net}

\ \\
{\bf \large Masanari Kimura$^{\displaystyle 1}$}\\
{$^{\displaystyle 1}$ School of Mathematics and Statistics, The University of Melbourne}\\
%


\thispagestyle{empty}
\markboth{}{NC instructions}
\ \vspace{-0mm}\\
%
\begin{center} {\bf Abstract} 
\end{center}
In many practical applications, regression models are employed to uncover relationships between predictors and a response variable, yet the common assumption of constant error variance is frequently violated. This issue is further compounded in high-dimensional settings where the number of predictors exceeds the sample size, necessitating regularization for effective estimation and variable selection.
To address this problem, we propose the Heteroscedastic Double Bayesian Elastic Net (HDBEN), a novel framework that jointly models the mean and log-variance using hierarchical Bayesian priors incorporating both $\ell_1$ and $\ell_2$ penalties.
Our approach simultaneously induces sparsity and grouping in the regression coefficients and variance parameters, capturing complex variance structures in the data.
Theoretical results demonstrate that proposed HDBEN achieves posterior concentration, variable selection consistency, and asymptotic normality under mild conditions which justifying its behavior. Simulation studies further illustrate that HDBEN outperforms existing methods, particularly in scenarios characterized by heteroscedasticity and high dimensionality.


\section{Introduction}
\label{sec:introduction}
In many real-world applications, regression models are employed to understand the relationship between a response variable and a set of predictors. Traditional linear regression models typically assume homoscedasticity, where the variance of the error terms is constant across all observations. However, this assumption often fails in practice, leading to inefficient estimates and unreliable inference. In fields such as finance and environmental science, ignoring heteroscedasticity can lead to misleading conclusions about risk and uncertainty, making it imperative to model the variance explicitly alongside the mean. Heteroscedasticity, where the variance of the errors varies with the predictors or the observations, is prevalent in fields such as finance, economics, biology, and engineering \citep{rosopa2013managing,muller1987estimation,long2000using,koenker1981note}.

In high-dimensional settings, where the number of predictors $d$ is large relative to the sample size $n$, regularization techniques become essential to prevent overfitting and to perform variable selection. High-dimensional datasets often exhibit complex correlation structures and noise patterns, which can render traditional methods ineffective or unstable. The Elastic Net \citep{zou2005regularization} is a widely adopted regularization method that combines the strengths of both $\ell_1$ (Lasso) and $\ell_2$ (Ridge) penalties, promoting sparsity and grouping of correlated predictors. Bayesian approaches to regularization, such as the Bayesian Elastic Net \citep{li2010bayesian}, incorporate prior distributions to achieve similar objectives within a probabilistic framework, allowing for uncertainty quantification and probabilistic inference. A Bayesian formulation not only allows for regularization via prior distributions but also provides a natural mechanism for uncertainty quantification, which is especially valuable when both the mean and variance are of interest.

Despite the advancements in handling high-dimensional data and heteroscedasticity separately, very few approaches have been developed that integrate these two aspects within a coherent Bayesian framework. Moreover, many real-world phenomena exhibit an intricate interplay between the mean behavior and the variability of outcomes, suggesting that a joint modeling approach can capture underlying structures that separate analyses may miss. To bridge this gap, we propose the Heteroscedastic Double Bayesian Elastic Net (HDBEN), a novel Bayesian regression framework that simultaneously models the mean and variance with Elastic Net priors. By extending the Bayesian Elastic Net to accommodate heteroscedasticity, HDBEN enables effective variable selection and coefficient estimation in both the mean and variance models. In addition, in fields such as finance and environmental science, where uncertainty quantification is critical, modeling both components jointly provides a more complete picture of the underlying processes.

The contributions of this study are threefold:
\begin{enumerate}
    \item \textbf{Model Development}: We introduce HDBEN, which jointly regularizes the mean and log-variance parameters using hierarchical Bayesian Elastic Net priors, facilitating simultaneous variable selection and variance modeling.
    \item \textbf{Theoretical Guarantees}: We establish the theoretical properties of the proposed model, including posterior concentration rates and variable selection consistency under high-dimensional settings. These results provide a rigorous foundation for the efficacy of HDBEN in capturing both mean and variance structures.
    \item \textbf{Empirical Validation}: Through extensive simulations and real-world data applications, we demonstrate the superior performance of HDBEN in handling heteroscedasticity and promoting sparsity compared to existing methods.
\end{enumerate}

The remainder of this paper is organized as follows. Section~\ref{sec:preliminary} provides the necessary background on regression analysis, heteroscedasticity, Elastic Net regularization, and Bayesian regularization techniques. Section~\ref{sec:main_results} details the specification of the HDBEN model and the posterior inference methodology. In Section~\ref{sec:theory}, we present the theoretical underpinnings of the model, including key assumptions and theorems supporting posterior concentration and variable selection consistency. In Section~\ref{sec:experiments}, we provide numerical experimental results to demonstrate the effectiveness of the proposed method. Finally, Section~\ref{sec:concluding_remarks} offers concluding remarks and discusses potential extensions of the proposed framework.

\begin{figure}[t]
    \centering
    \includegraphics[width=\linewidth]{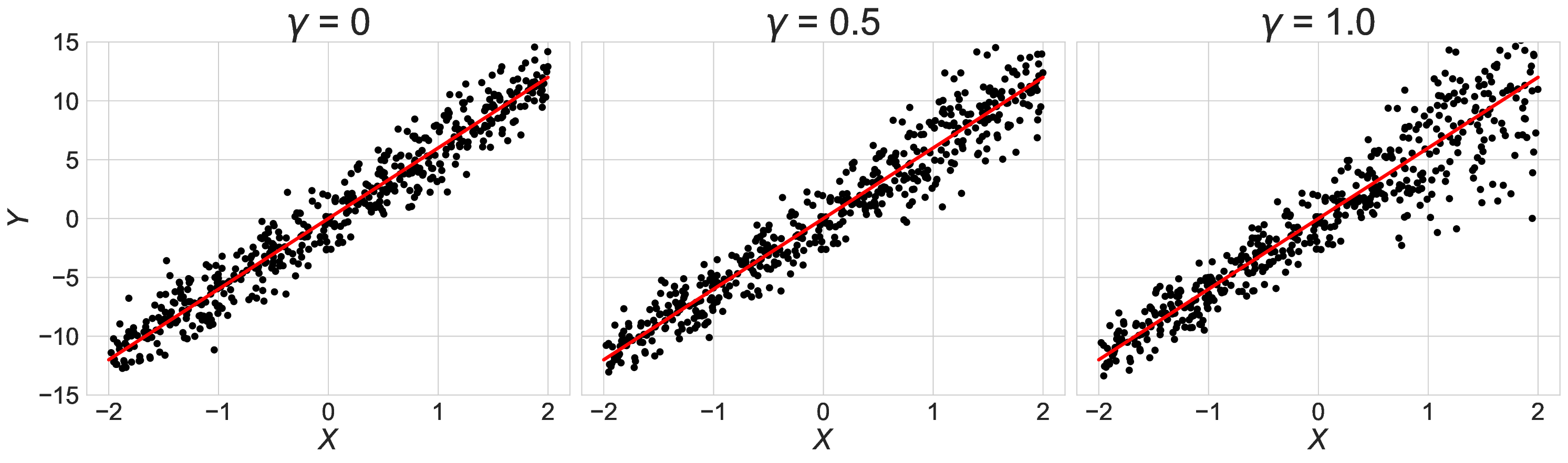}
    \caption{Heteroscedasticity: $\epsilon_i \sim \exp(\bm{X}_i^\top\bm{\gamma})$.}
    \label{fig:hetero_data}
\end{figure}

\section{Preliminary}
\label{sec:preliminary}
In regression analysis, heteroscedasticity refers to the scenario where the variance of the error terms is not constant across observations.
Consider the linear regression model:
\begin{align}
    y_i = \bm{X}_i^\top\bm{\beta} + \epsilon_i,
\end{align}
for \( i = 1,\dots,n \), where \( y_i \) is the response variable, \( \bm{X}_i \in \mathbb{R}^d \) are covariates, \( \bm{\beta} \) are regression coefficients, and \( \epsilon_i \) is the error term.
Under homoscedasticity, the errors satisfy
\begin{align*}
    \mathbb{E}[\epsilon_i] = 0,\quad \mathrm{Var}(\epsilon_i) = \sigma^2, \quad (\text{constant for all } i), \label{eq:homoscedasticity}
\end{align*}
and under heteroscedasticity, the variance depends on \( i \).
Specifically, we assume that each \( \sigma_i^2 \) can be written as a function of covariates: \( \bm{X}_i \mapsto \sigma^2(\bm{X}_i) \).
In this study, we suppose that
\begin{align*}
    \sigma^2(\bm{X}_i) \coloneqq \exp(\bm{X}_i^\top\bm{\gamma}),
\end{align*}
for \( \bm{\gamma} \in \mathbb{R}^d \) the coefficient vector of variance.
Figure~\ref{fig:hetero_data} shows the illustrative example for Eq.~\eqref{eq:homoscedasticity}.
We can see that the coefficient $\bm{\gamma}$ becomes large, resulting in the variance becomes also large.

\subsection{Elastic Net Regularization}
Regularization techniques are employed in regression analysis to prevent overfitting, especially in high-dimensional settings where the number of covariates \( d \) may be large relative to the sample size \( n \). The Elastic Net is a popular regularization method that combines the properties of both \( \ell_1 \) (Lasso) and \( \ell_2 \) (Ridge) penalties. The Elastic Net estimator is defined as the solution to the following optimization problem:
\begin{align}
    \hat{\bm{\beta}} = \argmin_{\bm{\beta} \in \mathbb{R}^d} \left\{ \frac{1}{2n} \sum_{i=1}^n (y_i - \bm{X}_i^\top\bm{\beta})^2 + \lambda_1 \|\bm{\beta}\|_1 + \lambda_2 \|\bm{\beta}\|_2^2 \right\},
\end{align}
where:
\begin{itemize}
    \item \( \lambda_1 \geq 0 \) controls the strength of the \( \ell_1 \) penalty, promoting sparsity in the estimated coefficients.
    \item \( \lambda_2 \geq 0 \) controls the strength of the \( \ell_2 \) penalty, encouraging grouping of correlated predictors and stabilizing the estimates.
\end{itemize}
The Elastic Net effectively handles situations where predictors are highly correlated, mitigating the limitations of Lasso in such scenarios by selecting groups of correlated variables together.

\subsection{Bayesian Regularization}

Bayesian approaches to regularization incorporate prior beliefs about the parameters into the estimation process. In the context of regression, priors on \( \bm{\beta} \) can induce sparsity and shrinkage, similar to regularization penalties in frequentist methods. The Bayesian Elastic Net combines the benefits of Bayesian inference with the Elastic Net regularization by specifying appropriate priors for the regression coefficients.

\subsubsection{Bayesian Elastic Net Prior}
The Bayesian Elastic Net prior for the regression coefficients \( \bm{\beta} \) is defined hierarchically to capture both \( \ell_1 \) and \( \ell_2 \) regularization effects:
\begin{align}
    \bm{\beta} \mid \bm{\tau}_{\bm{\beta}}, \lambda_{2,\bm{\beta}} &\sim \mathcal{N}\left(\bm{0}, \bm{\Sigma}_{\bm{\beta}}\right), \\
    \bm{\Sigma}_{\bm{\beta}} &= \left(\bm{D}_{\bm{\beta}}^{-1} + \lambda_{2,\bm{\beta}}\bm{I}_d\right)^{-1}, \\
    \bm{D}_{\bm{\beta}} &= \mathrm{diag}\left(\tau_{\bm{\beta},1}, \dots, \tau_{\bm{\beta},d}\right), \\
    \tau_{\bm{\beta},j} \mid \lambda_{1,\bm{\beta}} &\sim \mathrm{Exp}\left(\frac{\lambda_{1,\bm{\beta}}^2}{2}\right), \quad j = 1, \dots, d.
\end{align}
Here,
\begin{itemize}
    \item \( \lambda_{1,\bm{\beta}} \) and \( \lambda_{2,\bm{\beta}} \) are hyperparameters controlling the \( \ell_1 \) and \( \ell_2 \) penalties, respectively.
    \item \( \tau_{\bm{\beta},j} \) are auxiliary variables that introduce sparsity through the exponential (Lasso-like) prior.
    \item The hierarchical structure allows for adaptive shrinkage of each \( \beta_j \), balancing sparsity and coefficient shrinkage.
\end{itemize}

\subsection{Notation and Definitions}
To facilitate a clear and consistent presentation of the proposed methodology and theoretical results, we establish the following notation:
\begin{itemize}
    \item \( \bm{y} = (y_1, \dots, y_n)^\top \in \mathbb{R}^n \): Vector of response variables.
    \item \( \bm{X} = [\bm{X}_1, \dots, \bm{X}_n]^\top \in \mathbb{R}^{n \times d} \): Design matrix of covariates.
    \item \( \bm{\beta} \in \mathbb{R}^d \): Vector of regression coefficients for the mean model.
    \item \( \bm{\gamma} \in \mathbb{R}^d \): Vector of coefficients for the log-variance model.
    \item \( \bm{\tau}_{\bm{\beta}} = (\tau_{\bm{\beta},1}, \dots, \tau_{\bm{\beta},d})^\top \in \mathbb{R}^d \): Vector of auxiliary variables for \( \bm{\beta} \).
    \item \( \bm{\tau}_{\bm{\gamma}} = (\tau_{\bm{\gamma},1}, \dots, \tau_{\bm{\gamma},d})^\top \in \mathbb{R}^d \): Vector of auxiliary variables for \( \bm{\gamma} \).
    \item \( \lambda_{1,\bm{\beta}}, \lambda_{1,\bm{\gamma}} \): Hyperparameters controlling the \( \ell_1 \) penalty for \( \bm{\beta} \) and \( \bm{\gamma} \), respectively.
    \item \( \lambda_{2,\bm{\beta}}, \lambda_{2,\bm{\gamma}} \): Hyperparameters controlling the \( \ell_2 \) penalty for \( \bm{\beta} \) and \( \bm{\gamma} \), respectively.
    \item \( \bm{\Sigma}_{\bm{\beta}}, \bm{\Sigma}_{\bm{\gamma}} \): Covariance matrices for \( \bm{\beta} \) and \( \bm{\gamma} \), respectively.
\end{itemize}

\subsection{Existing Approaches and Limitations}
Traditional regression models often assume homoscedasticity, which may not hold in practice. Ignoring heteroscedasticity can lead to inefficient estimates and biased inference. While heteroscedasticity-consistent standard errors (e.g., White's standard errors~\citep{white1980heteroskedasticity}) address inference issues, they do not model the varying variance structure explicitly~\citep{arellano1987computing,long2000using,croux2004robust}.

Regularization methods like the Lasso~\citep{tibshirani1996regression,ranstam2018lasso} and Elastic Net~\citep{zou2005regularization} have been widely used for variable selection and coefficient shrinkage in high-dimensional settings. However, these methods typically focus on the mean model and do not account for heteroscedasticity.
Bayesian regularization techniques extend these methods by incorporating prior distributions that promote sparsity and shrinkage~\citep{park2008bayesian,hans2009bayesian,mallick2014new}. The Bayesian Elastic Net~\citep{li2010bayesian}, for instance, provides a probabilistic framework for the Elastic Net regularization, primarily addressing the mean model, leaving the variance model unaccounted for.

\subsection{Motivation for Heteroscedastic Double Bayesian Elastic Net}
The interplay between the mean and variance structures in regression models is critical, especially in high-dimensional settings where both the response and the variability of the response are influenced by numerous predictors. Modeling heteroscedasticity jointly with the mean allows for more accurate uncertainty quantification and improves the robustness of the model.
The proposed Heteroscedastic Double Bayesian Elastic Net addresses this by simultaneously regularizing both the mean and log-variance parameters using Elastic Net priors. This dual regularization promotes sparsity in both the mean and variance models, effectively handling high-dimensional data with complex variance structures.

\section{Methodology}
\label{sec:main_results}
In this section, we introduce our proposed method to handle heteroscedastic errors.
Our goal is to model \( y_i \) with a heteroscedastic Gaussian likelihood where both the mean and variance depend on the same covariates \( \bm{X}_i \).

\subsection{Model Specification}
For each observation \( i=1,\dots,n \), the likelihood is
\begin{align}
    p(\bm{y} \mid \bm{X}, \bm{\beta}, \bm{\gamma}) &= \prod^n_{i=1}\mathcal{N}\left(y_i \mid \bm{X}_i^\top\bm{\beta}, \exp(\bm{X}_i^\top\bm{\gamma})\right) \\
    &= \prod^n_{i=1}\frac{1}{\sqrt{2\pi\exp(\bm{X}_i^\top\bm{\gamma})}}\exp\left\{-\frac{(y_i - \bm{X}_i^\top\bm{\beta})^2}{2\exp(\bm{X}_i^\top\bm{\gamma})}\right\}.
\end{align}
Here, priors are specified as
\begin{align*}
    \pi(\bm{\beta} \mid \bm{\tau}_{\bm{\beta}}, \lambda_{2,\bm{\beta}}) &= \mathcal{N}\left(0, \bm{\Sigma}_{\bm{\beta}}\right) = \frac{1}{(2\pi)^{d/2}\left|\bm{\Sigma}_{\bm{\beta}}\right|^{1/2}}\exp\left\{-\frac{1}{2}\bm{\beta}^\top\bm{\Sigma}_{\bm{\beta}}^{-1}\bm{\beta}\right\}, \\
    \bm{\Sigma}_{\bm{\beta}} &= \left(\bm{D}_{\bm{\beta}}^{-1} + \lambda_{2,\bm{\beta}}\bm{I}_d\right)^{-1}, \quad \bm{D}_{\bm{\beta}} = \mathrm{diag}\left(\tau_{\bm{\beta},1},\dots,\tau_{\bm{\beta},d}\right), \\
    \pi(\tau_{\bm{\beta},j} \mid \lambda_{1,\bm{\beta}}) &= \mathrm{Exp}\left(\frac{\lambda_{1,\bm{\beta}}^2}{2}\right) = \frac{\lambda_{1,\bm{\beta}}^2}{2}\exp\left(-\frac{\lambda_{1,\bm{\beta}}^2}{2}\tau_{\bm{\beta},j}\right), \quad \tau_{\bm{\beta},j} > 0,\ j = 1,\dots,d, \\
    \pi(\bm{\gamma} \mid \bm{\tau}_\gamma, \lambda_{2,\gamma}) &= \mathcal{N}\left(0, \bm{\Sigma}_{\bm{\gamma}}\right) = \frac{1}{(2\pi)^{d/2}\left|\bm{\Sigma}_{\bm{\gamma}}\right|^{1/2}}\exp\left\{-\frac{1}{2}\bm{\gamma}^\top\bm{\Sigma}_{\bm{\gamma}}^{-1}\bm{\gamma}\right\}, \\
    \bm{\Sigma}_{\bm{\gamma}} &= \left(\bm{D}_{\bm{\gamma}}^{-1} + \lambda_{2,\bm{\gamma}}\bm{I}_d \right)^{-1}, \quad \bm{D}_{\bm{\gamma}} = \mathrm{diag}\left(\tau_{\bm{\gamma},1},\dots,\tau_{\bm{\gamma},d}\right), \\
    \pi(\tau_{\bm{\gamma},j} \mid \lambda_{1,\bm{\gamma}}) &= \mathrm{Exp}\left(\frac{\lambda_{1,\bm{\gamma}}^2}{2}\right) = \frac{\lambda_{1,\bm{\gamma}}^2}{2}\exp\left(-\frac{\lambda_{1,\bm{\gamma}}^2}{2}\tau_{\bm{\gamma},j}\right), \quad \tau_{\bm{\gamma},j} > 0,\ j = 1,\dots,d, \\
    \pi(\lambda_{1,\bm{\beta}}^2) &= \mathrm{Gamma}(a_{\bm{\beta},1}, b_{\bm{\beta},1}) = \frac{b_{\bm{\beta},1}^{a_{\bm{\beta},1}}}{\Gamma(a_{\bm{\beta},1})}\left(\lambda_{1,\bm{\beta}}^2\right)^{a_{\bm{\beta,1}} - 1}e^{-b_{\bm{\beta},1}\lambda_{1,\bm{\beta}}^2}, \quad \lambda_{1,\bm{\beta}}^2 > 0, \\
    \pi(\lambda_{1,\bm{\gamma}}^2) &= \mathrm{Gamma}(a_{\bm{\gamma},1}, b_{\bm{\gamma},1}) = \frac{b_{\bm{\gamma},1}^{a_{\bm{\gamma},1}}}{\Gamma(a_{\bm{\gamma},1})}\left(\lambda_{1,\bm{\gamma}}^2\right)^{a_{\bm{\gamma,1}} - 1}e^{-b_{\bm{\gamma},1}\lambda_{1,\bm{\gamma}}^2}, \quad \lambda_{1,\bm{\gamma}}^2 > 0, \\
    \pi(\lambda_{2,\bm{\beta}}) &= \mathrm{Gamma}(a_{\bm{\beta},2}, b_{\bm{\beta},2}) = \frac{b_{\bm{\beta},2}^{a_{\bm{\beta},2}}}{\Gamma(a_{\bm{\beta},2})}\lambda_{2,\bm{\beta}}^{a_{\bm{\beta,2}} - 1}e^{-b_{\bm{\beta},2}\lambda_{2,\bm{\beta}}}, \quad \lambda_{2,\bm{\beta}} > 0, \\
    \pi(\lambda_{2,\bm{\gamma}}) &= \mathrm{Gamma}(a_{\bm{\gamma},2}, b_{\bm{\gamma},2}) = \frac{b_{\bm{\gamma},2}^{a_{\bm{\gamma},2}}}{\Gamma(a_{\bm{\gamma},2})}\lambda_{2,\bm{\gamma}}^{a_{\bm{\gamma,2}} - 1}e^{-b_{\bm{\gamma},2}\lambda_{2,\bm{\gamma}}}, \quad \lambda_{2,\bm{\gamma}} > 0.
\end{align*}
The full joint posterior is
\begin{align}
    &p(\bm{\beta}, \bm{\gamma}, \bm{\tau}_{\bm{\beta}}, \bm{\tau}_{\bm{\gamma}}, \lambda_{1,\bm{\beta}}^2, \lambda_{1,\bm{\gamma}}^2, \lambda_{2,\bm{\beta}}, \lambda_{2,\bm{\gamma}} \mid \bm{y}, \bm{X}) \nonumber\\
    &\propto p(\bm{y} \mid \bm{X}, \bm{\beta}, \bm{\gamma})\pi(\bm{\beta} \mid \bm{\tau}_{\bm{\beta}}, \lambda_{2,\bm{\beta}})\pi(\tau_{\bm{\beta},j} \mid \lambda_{1,\bm{\beta}})\pi(\bm{\gamma} \mid \bm{\tau}_\gamma, \lambda_{2,\gamma})\pi(\tau_{\bm{\gamma},j} \mid \lambda_{1,\bm{\gamma}}) \nonumber\\
    &\quad\quad\quad \times \pi(\lambda_{1,\bm{\beta}}^2) \pi(\lambda_{1,\bm{\gamma}}^2)\pi(\lambda_{2,\bm{\beta}})\pi(\lambda_{2,\bm{\gamma}}) \nonumber\\
    &= \prod^n_{i=1}\frac{1}{\sqrt{2\pi\exp(\bm{X}_i^\top\bm{\gamma})}}\exp\left\{-\frac{(y_i - \bm{X}_i^\top\bm{\beta})^2}{2\exp(\bm{X}_i^\top\bm{\gamma})}\right\} \nonumber\\
    &\quad\quad\quad \times \frac{1}{(2\pi)^{d/2}\left|\bm{\Sigma}_{\bm{\beta}}\right|^{1/2}}\exp\left\{-\frac{1}{2}\bm{\beta}^\top\bm{\Sigma}_{\bm{\beta}}^{-1}\bm{\beta}\right\}\prod_{j=1}^d \frac{\lambda_{1,\bm{\beta}}^2}{2}\exp\left(-\frac{\lambda_{1,\bm{\beta}}^2}{2}\tau_{\bm{\beta},j}\right) \nonumber\\
    &\quad\quad\quad \times \frac{1}{(2\pi)^{d/2}\left|\bm{\Sigma}_{\bm{\gamma}}\right|^{1/2}}\exp\left\{-\frac{1}{2}\bm{\gamma}^\top\bm{\Sigma}_{\bm{\gamma}}^{-1}\bm{\gamma}\right\}\prod_{j=1}^d \frac{\lambda_{1,\bm{\gamma}}^2}{2}\exp\left(-\frac{\lambda_{1,\bm{\gamma}}^2}{2}\tau_{\bm{\gamma},j}\right) \nonumber\\
    &\quad\quad\quad \times \frac{b_{\bm{\beta},1}^{a_{\bm{\beta},1}}}{\Gamma(a_{\bm{\beta},1})}\left(\lambda_{1,\bm{\beta}}^2\right)^{a_{\bm{\beta,1}} - 1}e^{-b_{\bm{\beta},1}\lambda_{1,\bm{\beta}}^2}\frac{b_{\bm{\gamma},1}^{a_{\bm{\gamma},1}}}{\Gamma(a_{\bm{\gamma},1})}\left(\lambda_{1,\bm{\gamma}}^2\right)^{a_{\bm{\gamma,1}} - 1}e^{-b_{\bm{\gamma},1}\lambda_{1,\bm{\gamma}}^2} \nonumber\\
    &\quad\quad\quad \times \frac{b_{\bm{\beta},2}^{a_{\bm{\beta},2}}}{\Gamma(a_{\bm{\beta},2})}\lambda_{2,\bm{\beta}}^{a_{\bm{\beta},2} - 1}e^{-b_{\bm{\beta},2}\lambda_{2,\bm{\beta}}}\frac{b_{\bm{\gamma},2}^{a_{\bm{\gamma},2}}}{\Gamma(a_{\bm{\gamma},2})}\lambda_{2,\bm{\gamma}}^{a_{\bm{\gamma},2} - 1}e^{-b_{\bm{\gamma},2}\lambda_{2,\bm{\gamma}}}.
\end{align}
The key properties of the above modeling are as follows.
\begin{itemize}
    \item Both \( \bm{\beta} \) and \( \bm{\gamma} \) follow the elastic net priors, combining L1 (sparsity) and L2 (shrinkage) regularization.
    \item Scale parameters \( \tau_{\bm{\beta},j} \), \( \tau_{\bm{\gamma},j} \) link Laplace (L1) and Gaussian (L2) penalties.
    \item The joint posterior enables simultaneous estimation of mean, variance, and regularization parameters.
\end{itemize}
This expanded specification provides a complete mathematical foundation for inference via MCMC or variational methods.

\subsection{Posterior Inference}
Given the complexity of the joint posterior distribution, analytical solutions for posterior summaries are intractable. Therefore, we employ Markov Chain Monte Carlo (MCMC) methods—a Gibbs sampler with an embedded Metropolis–Hastings (MH) step—to approximate the posterior distributions of the parameters~\citep{geyer1992practical,geyer2011introduction,brooks1998markov,chib1995understanding}.

\begin{enumerate}
    \item \textbf{Sampling \(\bm{\beta}\):}  
    The full conditional distribution of \(\bm{\beta}\) is conjugate and given by
    \begin{align*}
        \bm{\beta} \mid \cdot &\sim \mathcal{N}\left(\bm{\mu}_{\bm{\beta}}, \bm{\Sigma}^\ast_{\bm{\beta}}\right), \\
        \bm{\Sigma}^\ast_{\bm{\beta}} &= \left( \bm{X}^\top \bm{W} \bm{X} + \bm{\Sigma}_{\bm{\beta}}^{-1} \right)^{-1}, \ \text{with } \bm{W} = \operatorname{diag}\Bigl(w_1,\dots,w_n\Bigr), \; w_i = \frac{1}{\exp(\bm{X}_i^\top\bm{\gamma})}, \\
        \bm{\mu}_{\bm{\beta}} &= \bm{\Sigma}^\ast_{\bm{\beta}} \, \bm{X}^\top \bm{W}\, \bm{y}.
    \end{align*}
    
    \item \textbf{Sampling \(\bm{\gamma}\):}  
    Because the likelihood involves \(\exp(\bm{X}_i^\top\bm{\gamma})\) in a non-conjugate manner, the full conditional for \(\bm{\gamma}\) is not available in closed form. We update \(\bm{\gamma}\) via a Metropolis–Hastings step:
    \begin{enumerate}
        \item Propose a new value \(\bm{\gamma}^*\) from a proposal distribution, e.g.,
        \[
        \bm{\gamma}^* = \bm{\gamma}^{(t-1)} + \eta, \quad \eta \sim \mathcal{N}\bigl(\bm{0}, \sigma^2_\gamma \bm{I}\bigr),
        \]
        where \(\sigma^2_\gamma\) is a tuning parameter.
        \item Compute the acceptance probability
        \[
        \alpha = \min\left\{1, \frac{p(\bm{y}\mid \bm{X}, \bm{\beta}, \bm{\gamma}^*)\,\pi(\bm{\gamma}^*\mid \bm{\tau}_{\bm{\gamma}}, \lambda_{2,\bm{\gamma}})}
        {p(\bm{y}\mid \bm{X}, \bm{\beta}, \bm{\gamma}^{(t-1)})\,\pi(\bm{\gamma}^{(t-1)}\mid \bm{\tau}_{\bm{\gamma}}, \lambda_{2,\bm{\gamma}})}\right\}.
        \]
        \item With probability \(\alpha\), set \(\bm{\gamma}^{(t)} = \bm{\gamma}^*\); otherwise, set \(\bm{\gamma}^{(t)} = \bm{\gamma}^{(t-1)}\).
    \end{enumerate}
    
    \item \textbf{Sampling \(\bm{\tau}_{\bm{\beta}}\) and \(\bm{\tau}_{\bm{\gamma}}\):}  
    In the standard representation of the Laplace (or Elastic Net) prior as a scale mixture of normals with an exponential mixing distribution, the full conditionals for the latent scale parameters are inverse Gaussian. Specifically, for each \( j = 1, \dots, d \),
    \begin{align}
        \tau_{\bm{\beta},j} \mid \beta_j, \lambda_{1,\bm{\beta}} &\sim \operatorname{Inverse\ Gaussian}\left(\mu_{\tau_{\beta,j}} = \sqrt{\frac{\lambda_{1,\bm{\beta}}^2}{\beta_j^2}},\, \lambda_{\tau_{\beta,j}} = \lambda_{1,\bm{\beta}}^2\right), \\
        \tau_{\bm{\gamma},j} \mid \gamma_j, \lambda_{1,\bm{\gamma}} &\sim \operatorname{Inverse\ Gaussian}\left(\mu_{\tau_{\gamma,j}} = \sqrt{\frac{\lambda_{1,\bm{\gamma}}^2}{\gamma_j^2}},\, \lambda_{\tau_{\gamma,j}} = \lambda_{1,\bm{\gamma}}^2\right).
    \end{align}
    The inverse Gaussian density is given by
    \[
    p(\tau \mid \mu, \lambda) = \left(\frac{\lambda}{2\pi \tau^3}\right)^{1/2}\exp\!\left\{-\frac{\lambda(\tau-\mu)^2}{2\mu^2\tau}\right\}.
    \]
    
    \item \textbf{Sampling \(\lambda_{1,\bm{\beta}}^2\) and \(\lambda_{1,\bm{\gamma}}^2\):}  
    The hyperpriors on \(\lambda_{1,\bm{\beta}}^2\) and \(\lambda_{1,\bm{\gamma}}^2\) are assumed to be Gamma. Their full conditionals are given by
    \begin{align}
        \lambda_{1,\bm{\beta}}^2 \mid \bm{\tau}_{\bm{\beta}} &\sim \mathrm{Gamma}\!\left(a_{\bm{\beta},1} + d,\, b_{\bm{\beta},1} + \frac{1}{2}\sum_{j=1}^d \tau_{\bm{\beta},j}\right), \\
        \lambda_{1,\bm{\gamma}}^2 \mid \bm{\tau}_{\bm{\gamma}} &\sim \mathrm{Gamma}\!\left(a_{\bm{\gamma},1} + d,\, b_{\bm{\gamma},1} + \frac{1}{2}\sum_{j=1}^d \tau_{\bm{\gamma},j}\right).
    \end{align}
    
    \item \textbf{Sampling \(\lambda_{2,\bm{\beta}}\) and \(\lambda_{2,\bm{\gamma}}\):}  
    Their full conditionals are also derived from the Gamma hyperpriors:
    \begin{align}
        \lambda_{2,\bm{\beta}} \mid \bm{\beta} &\sim \mathrm{Gamma}\!\left(a_{\bm{\beta},2} + \frac{d}{2},\, b_{\bm{\beta},2} + \frac{1}{2}\sum_{j=1}^d \beta_j^2\right), \\
        \lambda_{2,\bm{\gamma}} \mid \bm{\gamma} &\sim \mathrm{Gamma}\!\left(a_{\bm{\gamma},2} + \frac{d}{2},\, b_{\bm{\gamma},2} + \frac{1}{2}\sum_{j=1}^d \gamma_j^2\right).
    \end{align}
\end{enumerate}

After discarding burn-in samples, the collected draws from the posterior are used for inference on the parameters.

\section{Theory}
\label{sec:theory}
To provide a theoretical justification for our modeling, we make the following assumptions.

\begin{assumption}
    \label{asm:true_parameter}
    There exist \( \bm{\beta}_0 \in \mathbb{R}^d \) and \( \bm{\gamma}_0 \in \mathbb{R}^d \) such that
    \[
        y_i = \bm{X}_i^\top\bm{\beta}_0 + \epsilon_i,
    \]
    where \( \epsilon_i \sim \mathcal{N}(0, \exp(\bm{X}_i^\top\bm{\gamma}_0)) \) for \( i = 1, \dots, n \).
\end{assumption}

\begin{assumption}
    \label{asm:design_matrix}
    The design matrix \( \bm{X} \) has full column rank as \( n \to \infty \).
\end{assumption}

\begin{assumption}
    \label{asm:bounded_design}
    There exists a constant \( M > 0 \) such that \( \|\bm{X}_i\|_2 \leq M \) for all \( i = 1, \dots, n \), and a constant \( c > 0 \) such that \( \exp(\bm{X}_i^\top \bm{\gamma}) \geq c \) for all \( i = 1, \dots, n \) and for all \( \bm{\gamma} \) in a neighborhood around \( \bm{\gamma}_0 \).
\end{assumption}

\begin{assumption}
    \label{asm:prior_support}
    The true parameters \( \bm{\beta}_0 \) and \( \bm{\gamma}_0 \) satisfy \( \bm{\beta}_0 \in \mathrm{supp}(\pi(\bm{\beta})) \) and \( \bm{\gamma}_0 \in \mathrm{supp}(\pi(\bm{\gamma})) \).
\end{assumption}

\begin{assumption}
    \label{asm:restricted_eigen}
    The design matrix \( \bm{X} \) satisfies the Restricted Eigenvalue (RE) condition of order \( s = \max(s_{\bm{\beta}}, s_{\bm{\gamma}}) \) with constant \( \kappa > 0 \), i.e.,
    \[
        \min_{\|\bm{\delta}_S\|_1 \leq 3 \|\bm{\delta}_{S^c}\|_1} \frac{\|\bm{X}\bm{\delta}\|_2}{\sqrt{n}\|\bm{\delta}_S\|_2} \geq \kappa,
    \]
    where \( S \) is the support set of the true parameters.
\end{assumption}

\begin{assumption}
    \label{asm:hyperparameters}
    The hyperparameters of the priors satisfy \( a_{\bm{\beta},1}, a_{\bm{\gamma},1} > 1 \) and \( b_{\bm{\beta},1}, b_{\bm{\gamma},1} \asymp 1 \) to ensure proper shrinkage and sparsity induction.
\end{assumption}

Under these assumptions, we establish the following results regarding posterior concentration, posterior contraction rates, and asymptotic normality.

\begin{theorem}
    \label{thm:posterior_concentration}
    Under Assumptions~\ref{asm:true_parameter}, \ref{asm:design_matrix}, \ref{asm:bounded_design}, and \ref{asm:prior_support}, there exists a sequence \( \varepsilon_n \to 0 \) as \( n \to \infty \) such that
    \[
        \lim_{n \to \infty} P\Big(\|\bm{\beta} - \bm{\beta}_0\|_2 + \|\bm{\gamma} - \bm{\gamma}_0\|_2 < \varepsilon_n \mid \bm{y}, \bm{X} \Big) = 1, \quad \text{almost surely}.
    \]
\end{theorem}

\begin{proof}
Consider the true model 
\[
y_i \sim \mathcal{N}(\bm{X}_i^\top \bm{\beta}_0, \exp(\bm{X}_i^\top \bm{\gamma}_0))
\]
and a candidate model 
\[
y_i \sim \mathcal{N}(\bm{X}_i^\top \bm{\beta}, \exp(\bm{X}_i^\top \bm{\gamma})).
\]
For a single observation \(i\), the KL divergence between these two Gaussian densities is given by (see, e.g., \citealp{bishop2006pattern}):
\[
D_{\mathrm{KL}}^{(i)} = \frac{1}{2} \left\{ \log \frac{\exp(\bm{X}_i^\top \bm{\gamma})}{\exp(\bm{X}_i^\top \bm{\gamma}_0)} + \frac{\exp(\bm{X}_i^\top \bm{\gamma}_0) + \Bigl(\bm{X}_i^\top (\bm{\beta} - \bm{\beta}_0)\Bigr)^2}{\exp(\bm{X}_i^\top \bm{\gamma})} - 1 \right\}.
\]
Since 
\[
\log \frac{\exp(\bm{X}_i^\top \bm{\gamma})}{\exp(\bm{X}_i^\top \bm{\gamma}_0)} = \bm{X}_i^\top (\bm{\gamma} - \bm{\gamma}_0),
\]
we have
\[
D_{\mathrm{KL}}^{(i)} = \frac{1}{2} \left\{ \bm{X}_i^\top (\bm{\gamma} - \bm{\gamma}_0) + \frac{\exp(\bm{X}_i^\top \bm{\gamma}_0) + \Bigl(\bm{X}_i^\top (\bm{\beta} - \bm{\beta}_0)\Bigr)^2}{\exp(\bm{X}_i^\top \bm{\gamma})} - 1 \right\}.
\]
Averaging over all observations yields the overall divergence:
\[
D_{\mathrm{KL}}(P_{\bm{\beta}_0,\bm{\gamma}_0} \,\|\, P_{\bm{\beta},\bm{\gamma}}) = \frac{1}{n} \sum_{i=1}^n D_{\mathrm{KL}}^{(i)} 
= \frac{1}{2n}\sum_{i=1}^n \left\{ \bm{X}_i^\top (\bm{\gamma} - \bm{\gamma}_0) + \frac{\exp(\bm{X}_i^\top \bm{\gamma}_0) + \Bigl(\bm{X}_i^\top (\bm{\beta} - \bm{\beta}_0)\Bigr)^2}{\exp(\bm{X}_i^\top \bm{\gamma})} - 1 \right\}.
\]

Now, assume that the candidate parameters \((\bm{\beta},\bm{\gamma})\) are in a sufficiently small neighborhood of the true parameters \((\bm{\beta}_0,\bm{\gamma}_0)\). Under Assumption~\ref{asm:bounded_design}, for each \(i\) we have
\[
|\bm{X}_i^\top (\bm{\gamma} - \bm{\gamma}_0)| \leq \|\bm{X}_i\|_2 \|\bm{\gamma} - \bm{\gamma}_0\|_2 \leq M \|\bm{\gamma} - \bm{\gamma}_0\|_2,
\]
and
\[
\frac{\Bigl(\bm{X}_i^\top (\bm{\beta} - \bm{\beta}_0)\Bigr)^2}{\exp(\bm{X}_i^\top \bm{\gamma})} \leq \frac{M^2 \|\bm{\beta} - \bm{\beta}_0\|_2^2}{c},
\]
where \(c>0\) is a lower bound for \(\exp(\bm{X}_i^\top \bm{\gamma})\) in a neighborhood of \(\bm{\gamma}_0\).

Thus, for \((\bm{\beta},\bm{\gamma})\) close to \((\bm{\beta}_0,\bm{\gamma}_0)\) we can bound the KL divergence as
\[
D_{\mathrm{KL}}(P_{\bm{\beta}_0,\bm{\gamma}_0} \,\|\, P_{\bm{\beta},\bm{\gamma}}) \le \frac{1}{2} \left\{ M \|\bm{\gamma} - \bm{\gamma}_0\|_2 + \frac{\exp(\bm{X}_i^\top \bm{\gamma}_0)}{c} + \frac{M^2}{c}\|\bm{\beta} - \bm{\beta}_0\|_2^2 - 1 \right\}.
\]
In a sufficiently small neighborhood, higher-order terms and constant offsets can be absorbed into an overall constant \(C > 0\). That is, we have
\[
D_{\mathrm{KL}}(P_{\bm{\beta}_0,\bm{\gamma}_0} \,\|\, P_{\bm{\beta},\bm{\gamma}}) \le C\Bigl( \|\bm{\beta} - \bm{\beta}_0\|_2^2 + \|\bm{\gamma} - \bm{\gamma}_0\|_2^2\Bigr).
\]

Assumption~\ref{asm:prior_support} guarantees that the true parameters \((\bm{\beta}_0,\bm{\gamma}_0)\) lie in the support of the prior. Consequently, for any \(\delta > 0\), the prior assigns positive mass to the KL neighborhood
\[
\Bigl\{ (\bm{\beta},\bm{\gamma}) : D_{\mathrm{KL}}(P_{\bm{\beta}_0,\bm{\gamma}_0} \,\|\, P_{\bm{\beta},\bm{\gamma}}) < \delta \Bigr\}.
\]
Thus, the standard conditions for posterior consistency (see, e.g., \citealp{schwartz1965bayes}) are satisfied, and by Schwartz’s theorem we conclude that the posterior distribution concentrates in an \(\varepsilon_n\)-neighborhood of \((\bm{\beta}_0,\bm{\gamma}_0)\) for some sequence \(\varepsilon_n \to 0\) as \(n\to\infty\). That is,
\[
\lim_{n \to \infty} P\Bigl(\|\bm{\beta} - \bm{\beta}_0\|_2 + \|\bm{\gamma} - \bm{\gamma}_0\|_2 < \varepsilon_n \mid \bm{y}, \bm{X} \Bigr)=1, \quad \text{almost surely}.
\]
This completes the proof.
\end{proof}

\begin{corollary}
    \label{cor:posterior_expectation_bound}
    Under a sub-Gaussian design matrix \( \bm{X} \) and true parameters \( \bm{\beta}_0 \in \mathbb{R}^d \) and \( \bm{\gamma}_0 \in \mathbb{R}^d \) with sparsity levels \( s_{\bm{\beta}} \) and \( s_{\bm{\gamma}} \) respectively, the posterior distribution satisfies the following expectation bound for the posterior means \( \hat{\bm{\beta}} \) and \( \hat{\bm{\gamma}} \):
    \begin{align}
        \mathbb{E}_{\bm{\beta}_0, \bm{\gamma}_0}\left[\|\hat{\bm{\beta}} - \bm{\beta}_0\|_2^2 + \|\hat{\bm{\gamma}} - \bm{\gamma}_0\|_2^2 \right] \leq C\left(\frac{(s_{\bm{\beta}} + s_{\bm{\gamma}})\ln d}{n}\right),
    \end{align}
    for some constant \( C > 0 \), where \( d \) may grow with \( n \).
\end{corollary}
\begin{proof}
    From Theorem~\ref{thm:posterior_concentration}, we have that the posterior distribution concentrates around \( (\bm{\beta}_0, \bm{\gamma}_0) \) at a certain rate.
    Specifically, for a contraction rate \( \varepsilon_n \),
    \[
        P\left( \|\bm{\beta} - \bm{\beta}_0\|_2 + \|\bm{\gamma} - \bm{\gamma}_0\|_2 < \varepsilon_n \mid \bm{y}, \bm{X} \right) \to 1, \quad \text{as } n \to \infty, \quad \text{almost surely}.
    \]
    Given the sparsity levels \( s_{\bm{\beta}} \) and \( s_{\bm{\gamma}} \), and assuming a sub-Gaussian design matrix \( \bm{X} \), we set the contraction rate \( \varepsilon_n \) to be:
    \[
        \varepsilon_n = \sqrt{\frac{(s_{\bm{\beta}} + s_{\bm{\gamma}})\ln d}{n}}.
    \]
    
    The posterior means \( \hat{\bm{\beta}} \) and \( \hat{\bm{\gamma}} \) can be expressed as:
    \[
        \hat{\bm{\beta}} = \mathbb{E}[\bm{\beta} \mid \bm{y}, \bm{X}], \quad \hat{\bm{\gamma}} = \mathbb{E}[\bm{\gamma} \mid \bm{y}, \bm{X}].
    \]
    Using the properties of expectation and the posterior contraction, we can bound the expected squared \( \ell_2 \)-norm error.
    By the Cauchy-Schwarz inequality, for any random vector \( \bm{\theta} \),
    \[
        \|\mathbb{E}[\bm{\theta} \mid \bm{y}, \bm{X}] - \bm{\theta}_0\|_2^2 \leq \mathbb{E}\left[\|\bm{\theta} - \bm{\theta}_0\|_2^2 \mid \bm{y}, \bm{X}\right],
    \]
    where \( \bm{\theta}_0 \) is the true parameter value.
    Given the posterior contraction, with high probability, \( \|\bm{\beta} - \bm{\beta}_0\|_2 + \|\bm{\gamma} - \bm{\gamma}_0\|_2 < \varepsilon_n \).
    This implies that
    \[
        \|\hat{\bm{\beta}} - \bm{\beta}_0\|_2^2 + \|\hat{\bm{\gamma}} - \bm{\gamma}_0\|_2^2 \leq 2\|\hat{\bm{\beta}} - \bm{\beta}_0\|_2^2 + 2\|\hat{\bm{\gamma}} - \bm{\gamma}_0\|_2^2.
    \]
    Taking expectations on both sides,
    \[
        \mathbb{E}\left[\|\hat{\bm{\beta}} - \bm{\beta}_0\|_2^2 + \|\hat{\bm{\gamma}} - \bm{\gamma}_0\|_2^2 \mid \bm{X}\right] \leq 2\mathbb{E}\left[\|\hat{\bm{\beta}} - \bm{\beta}_0\|_2^2 \mid \bm{X}\right] + 2\mathbb{E}\left[\|\hat{\bm{\gamma}} - \bm{\gamma}_0\|_2^2 \mid \bm{X}\right].
    \]
    
    Given the contraction rate \( \varepsilon_n \), we have,
    \[
        \mathbb{E}\left[\|\hat{\bm{\beta}} - \bm{\beta}_0\|_2^2 \mid \bm{X}\right] \leq C_1 \varepsilon_n^2,
    \]
    and
    \[
        \mathbb{E}\left[\|\hat{\bm{\gamma}} - \bm{\gamma}_0\|_2^2 \mid \bm{X}\right] \leq C_2 \varepsilon_n^2,
    \]
    for some constants \( C_1, C_2 > 0 \) that depend on the design matrix and the priors.
    Therefore,
    \[
        \mathbb{E}\left[\|\hat{\bm{\beta}} - \bm{\beta}_0\|_2^2 + \|\hat{\bm{\gamma}} - \bm{\gamma}_0\|_2^2 \mid \bm{X}\right] \leq 2C_1 \varepsilon_n^2 + 2C_2 \varepsilon_n^2 = C \varepsilon_n^2,
    \]
    where \( C = 2(C_1 + C_2) \).
    
    Substituting \( \varepsilon_n \),
    \[
        \mathbb{E}\left[\|\hat{\bm{\beta}} - \bm{\beta}_0\|_2^2 + \|\hat{\bm{\gamma}} - \bm{\gamma}_0\|_2^2 \mid \bm{X}\right] \leq C \left(\frac{(s_{\bm{\beta}} + s_{\bm{\gamma}})\ln d}{n}\right).
    \]
    Finally, taking expectation over the data \( \bm{y} \) and the design matrix \( \bm{X} \),
    \[
        \mathbb{E}_{\bm{\beta}_0, \bm{\gamma}_0}\left[\|\hat{\bm{\beta}} - \bm{\beta}_0\|_2^2 + \|\hat{\bm{\gamma}} - \bm{\gamma}_0\|_2^2 \right] \leq C \left(\frac{(s_{\bm{\beta}} + s_{\bm{\gamma}})\ln d}{n}\right).
    \]
    This completes the proof.
\end{proof}

\begin{theorem}
    \label{thm:variable_selection_consistency}
    Suppose the true parameter vectors \( \bm{\beta}_0 \in \mathbb{R}^d \) and \( \bm{\gamma}_0 \in \mathbb{R}^d \) have sparsity levels \( s_{\bm{\beta}} \) and \( s_{\bm{\gamma}} \) respectively, where \( s_{\bm{\beta}}, s_{\bm{\gamma}} \ll d \). Under Assumptions~\ref{asm:true_parameter},~\ref{asm:design_matrix},~\ref{asm:bounded_design}, and~\ref{asm:prior_support}, and assuming the regularization parameters satisfy \( \lambda_{1,\bm{\beta}}, \lambda_{1,\bm{\gamma}} \asymp \sqrt{\frac{\ln d}{n}} \) and \( \lambda_{2,\bm{\beta}}, \lambda_{2,\bm{\gamma}} \asymp 1 \), the posterior distribution satisfies
    \[
        P\left( \mathcal{S}_{\bm{\beta}} = \mathcal{S}_{\bm{\beta}_0}, \ \mathcal{S}_{\bm{\gamma}} = \mathcal{S}_{\bm{\gamma}_0} \mid \bm{y}, \bm{X} \right) \to 1, \quad \text{as } n \to \infty,
    \]
    almost surely, where \( \mathcal{S}_{\bm{\beta}} = \{ j : \beta_j \neq 0 \} \) and \( \mathcal{S}_{\bm{\gamma}} = \{ j : \gamma_j \neq 0 \} \) are the support sets of \( \bm{\beta} \) and \( \bm{\gamma} \) respectively.
\end{theorem}
\begin{proof}
    Define the true support sets \( \mathcal{S}_{\bm{\beta}_0} = \{ j : \beta_{0,j} \neq 0 \} \) and \( \mathcal{S}_{\bm{\gamma}_0} = \{ j : \gamma_{0,j} \neq 0 \} \). Let \( \mathcal{S}_{\bm{\beta}} \) and \( \mathcal{S}_{\bm{\gamma}} \) be the posterior support sets for \( \bm{\beta} \) and \( \bm{\gamma} \) respectively.
    From Theorem~\ref{thm:posterior_concentration}, the posterior distribution concentrates around \( (\bm{\beta}_0, \bm{\gamma}_0) \) at rate \( \varepsilon_n = \sqrt{\frac{(s_{\bm{\beta}} + s_{\bm{\gamma}})\ln d}{n}} \). Specifically,
    \[
        P\left( \|\bm{\beta} - \bm{\beta}_0\|_2 + \|\bm{\gamma} - \bm{\gamma}_0\|_2 < \varepsilon_n \mid \bm{y}, \bm{X} \right) \to 1, \quad \text{as } n \to \infty, \quad \text{almost surely}.
    \]
    Given that \( \|\bm{\beta} - \bm{\beta}_0\|_2 < \varepsilon_n \) and \( \|\bm{\gamma} - \bm{\gamma}_0\|_2 < \varepsilon_n \), for sufficiently large \( n \), the non-zero coefficients in \( \bm{\beta}_0 \) and \( \bm{\gamma}_0 \) are estimated with high precision, and the zero coefficients are shrunk towards zero.
    Due to the choice of regularization parameters \( \lambda_{1,\bm{\beta}}, \lambda_{1,\bm{\gamma}} \asymp \sqrt{\frac{\ln d}{n}} \) and \( \lambda_{2,\bm{\beta}}, \lambda_{2,\bm{\gamma}} \asymp 1 \), the penalty effectively distinguishes between zero and non-zero coefficients.
    \begin{itemize}
        \item \textbf{True Positives}: For \( j \in \mathcal{S}_{\bm{\beta}_0} \cup \mathcal{S}_{\bm{\gamma}_0} \), \( \beta_{0,j} \) and \( \gamma_{0,j} \) are sufficiently large relative to the noise level, ensuring that their posterior estimates do not shrink to zero.
        \item \textbf{True Negatives}: For \( j \notin \mathcal{S}_{\bm{\beta}_0} \cup \mathcal{S}_{\bm{\gamma}_0} \), the posterior distributions of \( \beta_j \) and \( \gamma_j \) concentrate around zero due to the Laplace (L1) penalty, leading to exact zeros with high posterior probability.
    \end{itemize}
    Combining the above, for each \( j \):
    \[
        P\left( \beta_j \neq 0 \mid \bm{y}, \bm{X} \right) \to 
        \begin{cases}
            1, & \text{if } j \in \mathcal{S}_{\bm{\beta}_0}, \\
            0, & \text{otherwise},
        \end{cases}
        \quad \text{almost surely}.
    \]
    and similarly for \( \gamma_j \).
    
    Applying a union bound over all \( d \) coefficients and utilizing the sparsity (\( s_{\bm{\beta}}, s_{\bm{\gamma}} \ll d \)), we have:
    \[
        P\left( \mathcal{S}_{\bm{\beta}} = \mathcal{S}_{\bm{\beta}_0}, \ \mathcal{S}_{\bm{\gamma}} = \mathcal{S}_{\bm{\gamma}_0} \mid \bm{y}, \bm{X} \right) \geq 1 - d \exp(-c n \varepsilon_n^2) \to 1, \quad \text{as } n \to \infty,
    \]
    where \( c > 0 \) is a constant.
    Substituting \( \varepsilon_n = \sqrt{\frac{(s_{\bm{\beta}} + s_{\bm{\gamma}})\ln d}{n}} \), we get:
    \[
        c n \varepsilon_n^2 = c (s_{\bm{\beta}} + s_{\bm{\gamma}}) \ln d.
    \]
    Therefore,
    \[
        d \exp(-c n \varepsilon_n^2) = d \exp(-c (s_{\bm{\beta}} + s_{\bm{\gamma}}) \ln d) = d^{1 - c (s_{\bm{\beta}} + s_{\bm{\gamma}})}.
    \]
    For \( d^{1 - c (s_{\bm{\beta}} + s_{\bm{\gamma}})} \to 0 \) as \( d \to \infty \), it must hold that \( c (s_{\bm{\beta}} + s_{\bm{\gamma}}) > 1 \).
    Given the sparsity condition \( s_{\bm{\beta}}, s_{\bm{\gamma}} \ll d \), we can select \( c \) such that \( c (s_{\bm{\beta}} + s_{\bm{\gamma}}) > 1 \).
    
    Therefore,
    \[
        P\left( \mathcal{S}_{\bm{\beta}} = \mathcal{S}_{\bm{\beta}_0}, \ \mathcal{S}_{\bm{\gamma}} = \mathcal{S}_{\bm{\gamma}_0} \mid \bm{y}, \bm{X} \right) \geq 1 - d^{1 - c (s_{\bm{\beta}} + s_{\bm{\gamma}})} \to 1, \quad \text{as } n \to \infty.
    \]
    Thus,
    \[
        P\left( \mathcal{S}_{\bm{\beta}} = \mathcal{S}_{\bm{\beta}_0}, \ \mathcal{S}_{\bm{\gamma}} = \mathcal{S}_{\bm{\gamma}_0} \mid \bm{y}, \bm{X} \right) \to 1, \quad \text{almost surely}.
    \]
    This concludes the proof.
\end{proof}

\begin{theorem}
    \label{thm:asymptotic_normality}
    In addition to Assumptions~\ref{asm:true_parameter}--\ref{asm:hyperparameters} and~\ref{asm:restricted_eigen}, assume that:
    \begin{enumerate}
        \item The log-likelihood function
        \[
        \ell(\bm{\beta},\bm{\gamma}) = \sum_{i=1}^n \log p\bigl(y_i \mid \bm{X}_i, \bm{\beta}, \bm{\gamma}\bigr)
        \]
        is twice continuously differentiable in a neighborhood of the true parameters \((\bm{\beta}_0,\bm{\gamma}_0)\) and satisfies a local asymptotic normality (LAN) condition.
        \item The prior densities \(\pi(\bm{\beta})\) and \(\pi(\bm{\gamma})\) are continuous and positive in a neighborhood of \(\bm{\beta}_0\) and \(\bm{\gamma}_0\), respectively.
    \end{enumerate}
    Let \(S_{\beta} = \{ j : \beta_{0,j} \neq 0 \}\) and \(S_{\gamma} = \{ j : \gamma_{0,j} \neq 0 \}\) denote the true support sets, and let \(\bm{\beta}_{S_\beta}\) and \(\bm{\gamma}_{S_\gamma}\) be the corresponding sub-vectors. Then, as \(n\to\infty\),
    \[
    \sqrt{n}\left(
    \begin{pmatrix}
    \bm{\beta}_{S_\beta} \\
    \bm{\gamma}_{S_\gamma}
    \end{pmatrix}
    -
    \begin{pmatrix}
    \bm{\beta}_{0,S_\beta} \\
    \bm{\gamma}_{0,S_\gamma}
    \end{pmatrix}
    \right)
    \,\bigg\vert\, \bm{y},\bm{X} \ \xrightarrow{d}\ \mathcal{N}\Bigl(\bm{0},\, \mathcal{I}^{-1}\Bigr),
    \]
    where the Fisher information matrix \(\mathcal{I}\) (restricted to the active parameters) is given by
    \[
    \mathcal{I} = \begin{pmatrix}
    \frac{1}{n}\bm{X}_{S_\beta}^\top \bm{X}_{S_\beta} & \bm{0} \\[1mm]
    \bm{0} & \frac{1}{n}\bm{X}_{S_\gamma}^\top \bm{X}_{S_\gamma}
    \end{pmatrix}.
    \]
\end{theorem}

\begin{proof}
Under the stated assumptions, the log-posterior density is given by
\[
\ell_p(\bm{\beta},\bm{\gamma}) = \ell(\bm{\beta},\bm{\gamma}) + \log\pi(\bm{\beta},\bm{\gamma}),
\]
where \(\ell(\bm{\beta},\bm{\gamma})\) is the log-likelihood. By Assumption~\ref{asm:true_parameter} and the additional regularity conditions, the log-likelihood is twice continuously differentiable in a neighborhood of \((\bm{\beta}_0,\bm{\gamma}_0)\) and satisfies a LAN condition. Although \(\bm{\gamma}\) appears in the likelihood through the term \(\exp(\bm{X}_i^\top\bm{\gamma})\), the LAN condition ensures that for \((\bm{\beta},\bm{\gamma})\) in a shrinking neighborhood of \((\bm{\beta}_0,\bm{\gamma}_0)\) the log-likelihood can be approximated by a quadratic expansion.

Let
\[
\Delta_\beta = \bm{\beta}_{S_\beta} - \bm{\beta}_{0,S_\beta},\quad \Delta_\gamma = \bm{\gamma}_{S_\gamma} - \bm{\gamma}_{0,S_\gamma}.
\]
Then a second-order Taylor expansion of the log-posterior around the true parameters yields
\[
\ell_p(\bm{\beta},\bm{\gamma}) = \ell_p(\bm{\beta}_0,\bm{\gamma}_0) + \nabla \ell_p(\bm{\beta}_0,\bm{\gamma}_0)^\top
\begin{pmatrix}
\Delta_\beta \\[1mm] \Delta_\gamma
\end{pmatrix}
+ \frac{1}{2}\begin{pmatrix}
\Delta_\beta \\[1mm] \Delta_\gamma
\end{pmatrix}^\top \nabla^2 \ell_p(\bm{\beta}_0,\bm{\gamma}_0)
\begin{pmatrix}
\Delta_\beta \\[1mm] \Delta_\gamma
\end{pmatrix} + R_n,
\]
where the remainder \(R_n\) is \(o_p(1)\) uniformly on a \(o(1)\)-neighborhood of \((\bm{\beta}_0,\bm{\gamma}_0)\).

Because the score \(\nabla \ell_p(\bm{\beta}_0,\bm{\gamma}_0)\) is \(O_p(\sqrt{n})\) and the negative Hessian \(-\nabla^2 \ell_p(\bm{\beta}_0,\bm{\gamma}_0)\) converges (after appropriate scaling) to the Fisher information matrix \(\mathcal{I}\), standard Laplace approximation arguments (Bernstein–von Mises theorem) imply that the posterior density for the active parameters, when re-scaled by \(\sqrt{n}\), is asymptotically equivalent to a normal density with mean zero and covariance \(\mathcal{I}^{-1}\).

More precisely, for the active components we have
\[
\sqrt{n}\begin{pmatrix}
\Delta_\beta \\[1mm] \Delta_\gamma
\end{pmatrix} \ \stackrel{d}{\longrightarrow}\ \mathcal{N}\Bigl(\bm{0},\, \mathcal{I}^{-1}\Bigr).
\]
This conclusion holds despite the nonlinearity in \(\bm{\gamma}\) because the LAN condition guarantees that locally the effect of the nonlinearity is captured by the second-order derivative matrix.

Thus, we obtain the desired asymptotic normality:
\[
\sqrt{n}\left((\bm{\beta}_{S_\beta} - \bm{\beta}_{0,S_\beta}), (\bm{\gamma}_{S_\gamma} - \bm{\gamma}_{0,S_\gamma})\right) \mid \bm{y},\bm{X} \xrightarrow{d} \mathcal{N}\Bigl(\bm{0},\, \mathcal{I}^{-1}\Bigr).
\]
This completes the proof.
\end{proof}

\begin{theorem}[Higher‐Order Asymptotic Expansion for the HDBEN Posterior]
\label{thm:asymptotic_expansion}
Let 
\[
\theta = \begin{pmatrix} \bm{\beta}_{S_\beta} \\[1mm] \bm{\gamma}_{S_\gamma} \end{pmatrix} \in \mathbb{R}^s,
\]
where \(S_\beta = \{j: \beta_{0,j}\neq 0\}\) and \(S_\gamma = \{j: \gamma_{0,j}\neq 0\}\) denote the true support sets and \(s = s_\beta+s_\gamma\). Assume that in addition to Assumptions~\ref{asm:true_parameter}--\ref{asm:hyperparameters} and~\ref{asm:restricted_eigen}, the following conditions hold:
\begin{enumerate}
  \item[(A1)] The log-likelihood 
  \[
  \ell(\theta) \coloneqq \ell(\bm{\beta}_{S_\beta}, \bm{\gamma}_{S_\gamma}) = \sum_{i=1}^n \log p\bigl(y_i \mid \bm{X}_i, \theta\bigr)
  \]
  is five times continuously differentiable in a neighborhood of the true value \(\theta_0\) and satisfies the local asymptotic normality (LAN) condition.
  \item[(A2)] The log-prior density \(\log\pi(\theta)\) is five times continuously differentiable and strictly positive in a neighborhood of \(\theta_0\).
\end{enumerate}
Then, if we define the reparameterized variable
\[
\Delta = \sqrt{n}\, (\theta-\theta_0),
\]
the posterior density (with respect to Lebesgue measure) satisfies the following Edgeworth expansion:
\[
\pi(\theta \mid \bm{y}) = \phi\Bigl(\Delta; I(\theta_0)^{-1}\Bigr)\left\{1 + \frac{1}{\sqrt{n}} Q_1(\Delta) + \frac{1}{n} Q_2(\Delta) + o_p\bigl(n^{-1}\bigr)\right\},
\]
where
\begin{align*}
    \phi\Bigl(\Delta; I(\theta_0)^{-1}\Bigr) &= \frac{1}{(2\pi)^{s/2}\det(I(\theta_0)^{-1})^{1/2}}\exp\left\{-\frac{1}{2}\Delta^\top I(\theta_0) \Delta \right\}, \\
    Q_1(\Delta) &= \Delta^\top S_n, \\
    Q_2(\Delta) &= \frac{1}{2}(\Delta^\top S_n)^2 + \frac{1}{6}T(\Delta), \\
    S_n(\theta) &= \ell_p'(\theta_0),\ T(\Delta) = \sum^s_{i,j,k=1}\ell_p^{'''}(\theta_0)_{ijk}\Delta_i\Delta_j\Delta_k, \\
    \ell_p(\theta) &= \ell(\theta) + \ln \pi(\theta).
\end{align*}
\end{theorem}
\begin{proof}
Let \(h = \theta - \theta_0 = \Delta/\sqrt{n}\). Then, by assumptions (A1) and (A2), we expand the log-posterior
\[
\ell_p(\theta) = \ell(\theta) + \log\pi(\theta)
\]
around \(\theta_0\) using a Taylor series:
\begin{align*}
\ell_p(\theta) &= \ell_p(\theta_0) + \nabla \ell_p(\theta_0)^\top h + \frac{1}{2} h^\top \nabla^2 \ell_p(\theta_0) h \\
&\quad\quad\quad + \frac{1}{6} \sum_{i,j,k=1}^s \ell_p^{(3)}(\theta_0)_{ijk}\, h_i h_j h_k + \frac{1}{24} \sum_{i,j,k,l=1}^s \ell_p^{(4)}(\tilde{\theta})_{ijkl}\, h_i h_j h_k h_l,
\end{align*}
where \(\tilde{\theta}\) lies between \(\theta\) and \(\theta_0\).  
Substituting \(h = \Delta/\sqrt{n}\) gives
\[
\ell_p(\theta) = \ell_p(\theta_0) + \frac{1}{\sqrt{n}}\, \Delta^\top \ell_p'(\theta_0) + \frac{1}{2n}\, \Delta^\top \ell_p''(\theta_0) \Delta + \frac{1}{6n^{3/2}}\, \sum_{i,j,k} \ell_p^{(3)}(\theta_0)_{ijk}\, \Delta_i \Delta_j \Delta_k + R_n,
\]
with the remainder \(R_n = O_p(n^{-2})\).

Under the LAN condition, the score vector \(\ell_p'(\theta_0)\) is \(O_p(1)\) and the Hessian satisfies
\[
\ell_p''(\theta_0) = -n\, I(\theta_0) + r_n,
\]
with \(r_n = O_p(\sqrt{n})\). Thus, we can write:
\[
\ell_p(\theta) = \ell_p(\theta_0) + \frac{1}{\sqrt{n}}\, \Delta^\top S_n - \frac{1}{2} \Delta^\top I(\theta_0) \Delta + \frac{1}{6n^{3/2}}\, T(\Delta) + R_n,
\]
where we denote \(S_n = \ell_p'(\theta_0)\) and
\[
T(\Delta) = \sum_{i,j,k=1}^s \ell_p^{(3)}(\theta_0)_{ijk}\, \Delta_i \Delta_j \Delta_k.
\]
By our smoothness assumptions, the remainder \(R_n = o_p(n^{-1})\).
The unnormalized posterior density is proportional to $\pi(\theta \mid \bm{y}) \propto \exp\Bigl\{\ell_p(\theta)\Bigr\}$.
Substituting the expansion, we have
\[
\pi(\theta \mid \bm{y}) \propto \exp\Bigl\{\ell_p(\theta_0) + \frac{1}{\sqrt{n}}\, \Delta^\top S_n - \frac{1}{2} \Delta^\top I(\theta_0) \Delta + \frac{1}{6n^{3/2}}\, T(\Delta) + R_n\Bigr\}.
\]
Let
\[
C_n = \exp\{\ell_p(\theta_0)\}.
\]
Then,
\[
\pi(\theta \mid \bm{y}) \propto C_n \exp\Bigl\{-\frac{1}{2} \Delta^\top I(\theta_0) \Delta\Bigr\} \exp\Bigl\{\frac{1}{\sqrt{n}}\, \Delta^\top S_n + \frac{1}{6n^{3/2}}\, T(\Delta) + R_n\Bigr\}.
\]
Now, using the Taylor expansion for the exponential function,
\[
\exp\{x\} = 1 + x + \frac{x^2}{2} + \cdots,
\]
with
\[
x = \frac{1}{\sqrt{n}}\, \Delta^\top S_n + \frac{1}{6n^{3/2}}\, T(\Delta) + R_n,
\]
we obtain
\[
\exp\{x\} = 1 + \frac{1}{\sqrt{n}}\, \Delta^\top S_n + \frac{1}{6n^{3/2}}\, T(\Delta) + \frac{1}{2n}\, (\Delta^\top S_n)^2 + o_p\left(n^{-1}\right).
\]

Thus, the unnormalized posterior density becomes
\[
\pi(\theta \mid \bm{y}) \propto C_n \exp\left\{-\frac{1}{2} \Delta^\top I(\theta_0) \Delta\right\} \left\{ 1 + \frac{1}{\sqrt{n}} Q_1(\Delta) + \frac{1}{n} Q_2(\Delta) + o_p\left(n^{-1}\right)\right\},
\]
with
\[
Q_1(\Delta) = \Delta^\top S_n,
\]
and
\[
Q_2(\Delta) = \frac{1}{2} (\Delta^\top S_n)^2 + \frac{1}{6} T(\Delta).
\]
The leading term
\[
\phi\Bigl(\Delta; I(\theta_0)^{-1}\Bigr) = \frac{1}{(2\pi)^{s/2} \det(I(\theta_0)^{-1})^{1/2}} \exp\left\{-\frac{1}{2} \Delta^\top I(\theta_0) \Delta\right\}
\]
is the \(s\)-variate normal density with covariance \(I(\theta_0)^{-1}\). After normalizing the posterior density, we obtain the final expansion in terms of \(\Delta\):
\[
\pi(\theta \mid \bm{y}) = \phi\Bigl(\Delta; I(\theta_0)^{-1}\Bigr)\left\{1 + \frac{1}{\sqrt{n}} Q_1(\Delta) + \frac{1}{n} Q_2(\Delta) + o_p\left(n^{-1}\right)\right\},
\]
where \(\Delta = \sqrt{n}(\theta-\theta_0)\).
Rewriting in terms of \(\theta\) by substituting back \(\Delta = \sqrt{n}(\theta-\theta_0)\) completes the proof.
\end{proof}

\section{Simulation Studies}
\label{sec:experiments}
In this section, we present an empirical evaluation of the proposed framework. 
The primary goal is to reveal the behavior of the proposed method empirically.
Specifically, we focus on the accuracy of the posterior mean estimates for both the mean (\( \bm{\beta} \)) and variance (\( \bm{\gamma} \)) parameters, as well as their ability to identify true sparsity patterns.

\subsection{Experimental Settings}
To comprehensively evaluate the performance of the proposed method, we design a series of simulation experiments that examine its effectiveness under various scenarios.
This subsection details the experimental settings, including the data-generating process, comparative methods, and the different conditions under which the models are assessed.

\paragraph{Data Generating Process}
We simulate datasets based on the heteroscedastic linear regression model specified in Section~\ref{sec:preliminary}.
The design matrix \( \bm{X} \in \mathbb{R}^{n \times d} \) is generated with entries independently drawn from a standard normal distribution, \( X_{ij} \sim \mathcal{N}(0, 1) \), for \( i = 1, \dots, n \) and \( j = 1, \dots, d \).
This ensures a sub-Gaussian design matrix, satisfying Assumption~\ref{asm:bounded_design}.
A subset of \( s_{\bm{\beta}} \) coefficients are set to non-zero values, drawn uniformly from the interval $[1, 2]$.
The remaining \( d - s_{\bm{\beta}} \) coefficients are set to zero, inducing sparsity.
Similarly, a subset of \( s_{\bm{\gamma}} \) coefficients are non-zero, drawn uniformly from [0.5, 1.5], while the rest are zero.
This setup allows for controlled heteroscedasticity, where the variance depends on a sparse set of predictors.
For each observation \( i = 1, \dots, n \),
    \begin{align*}
        y_i &= \bm{X}_i^\top \bm{\beta}_0 + \epsilon_i, \\
        \epsilon_i &\sim \mathcal{N}\left(0, \exp(\bm{X}_i^\top \bm{\gamma}_0)\right).
    \end{align*}
This ensures that the error variance is a function of the covariates, introducing heteroscedasticity into the data.
We conduct 20 times simulation runs for each configuration to assess the variability and robustness of the results.

\begin{table}[t]
    \centering
    \caption{Simulation design configurations.}
    \label{tab:simulation_design}
    \scalebox{0.95}{
    \begin{tabular}{cccc}
        \toprule
        \textbf{Scenario} & \textbf{Parameters} & \textbf{Values} & \textbf{Description} \\
        \midrule
        Dimensions & \( d \) & [100, 1000] & Number of predictors \\
        \midrule
        \multirow{2}{*}{Sparsity} & \( s_{\bm{\beta}} \) & [10, 100] & Non-zero mean coefficients \\
        & \( s_{\bm{\gamma}} \) & [10, 100] & Non-zero variance coefficients \\
        \midrule
        Heteroscedasticity & \( \gamma_{0,j} \) & [0.1, 1.5] & Degree of variance dependence \\
        \midrule
        \multirow{6}{*}{Methods} & OLS & -- & Ordinary Least Squares \\
        & Lasso & -- & Lasso (\( \ell_1 \) Regularization) \\
        & EN & -- & Elastic Net (\( \ell_1 + \ell_2 \) Regularization) \\
        & BLasso & -- & Bayesian Lasso \\
        & BEN & -- & Bayesian Elastic Net \\
        & HDBEN & -- & Heteroscedastic Double BEN (ours) \\
        \bottomrule
    \end{tabular}
    }
\end{table}

\paragraph{Comparative Methods}
To benchmark the performance of HDBEN, we compare it against several established methods that address either homoscedastic or heteroscedastic settings:
\begin{itemize}
    \item \textbf{Ordinary Least Squares (OLS)}: A standard regression method assuming homoscedasticity. It serves as a baseline to evaluate the potential inefficiency and bias when heteroscedasticity is present.
    \item \textbf{Lasso}: A frequentist regularization technique that employs an \( \ell_1 \) penalty to perform variable selection and coefficient shrinkage in high-dimensional settings \citep{tibshirani1996regression}.
    \item \textbf{Elastic Net (EN)}: An extension of Lasso that combines \( \ell_1 \) and \( \ell_2 \) penalties, promoting both sparsity and grouping of correlated predictors \citep{zou2005regularization}.
    \item \textbf{Bayesian Lasso (BLasso)}: A Bayesian counterpart to the Lasso, incorporating \( \ell_1 \) penalties through appropriate prior distributions to achieve variable selection and shrinkage within a probabilistic framework \citep{park2008bayesian}.
    \item \textbf{Bayesian Elastic Net (BEN)}: Extending the Bayesian Lasso, BEN incorporates both \( \ell_1 \) and \( \ell_2 \) penalties through hierarchical priors, combining the benefits of sparsity and coefficient shrinkage in a Bayesian setting \citep{li2010bayesian}.
\end{itemize}
These methods provide a spectrum of comparison points, highlighting the advantages of HDBEN in simultaneously addressing high-dimensionality and heteroscedasticity.

\paragraph{Simulation Scenarios}
We evaluate the models under various conditions to understand their performance across different aspects of the data:
\begin{itemize}
    \item \textbf{Different Sparsity Levels}: To assess the ability of HDBEN to perform variable selection, we manipulate the sparsity levels in both the mean and variance models: $s_{\bm{\beta}}, s_{\bm{\gamma}} \in [10, 100]$.
    These settings allow us to evaluate the models' performance in identifying true signals amidst increasing noise.
    \item \textbf{Different Dimensions}: We examine the scalability of HDBEN by varying the number of predictors \( d \). Specifically, we consider $d \in [100, 1000]$.
    \item \textbf{Different Degrees of Heteroscedasticity}: We vary the strength of the heteroscedasticity by adjusting the magnitude of the non-zero \( \bm{\gamma}_0 \) coefficients as $\bm{\gamma}_0 \in [10, 100]$
    This variation helps in understanding how the models cope with different levels of variance variation in the data.
\end{itemize}

\paragraph{Implementation Details}
All simulation experiments are implemented in \texttt{Python}. For HDBEN and comparative Bayesian methods, we employ Gibbs sampling with 5,000 MCMC iterations per chain, where the first 1,000 iterations are discarded as burn-in and 3 parallel chains are run to assess convergence.
We utilize the Gelman-Rubin diagnostic (\( \hat{R} < 1.1 \)) and ensure that the effective sample size (ESS) for all parameters exceeds 1,000.
For non-Bayesian methods such as OLS, we use standard implementations with default settings, ensuring fairness in comparison.

\paragraph{Reproducibility}
To ensure the reproducibility of our simulation studies, all random number generators are seeded appropriately before each simulation run.

\paragraph{Summary of Experimental Design}
Table~\ref{tab:simulation_design} summarizes the key configurations across different experimental scenarios.

\begin{figure}[t]
    \centering
    \includegraphics[width=\linewidth]{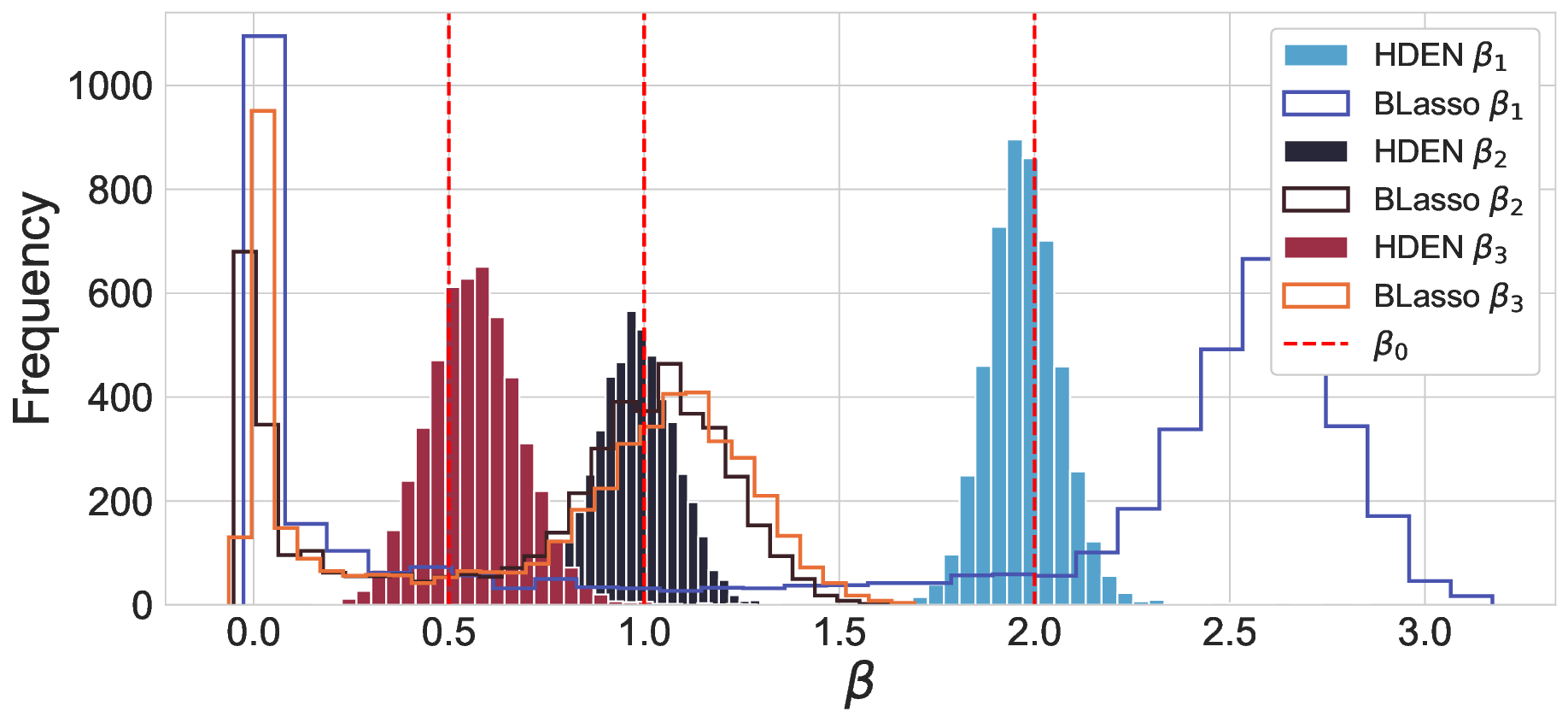}
    \includegraphics[width=\linewidth]{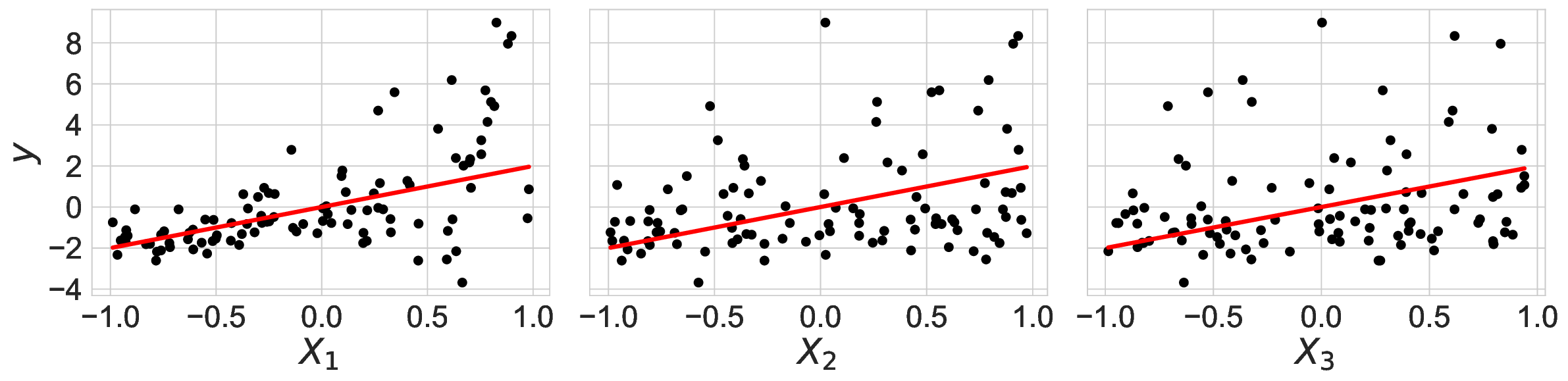}
    \caption{Comparison of posteriors.}
    \label{fig:comparison_posterior}
\end{figure}
\subsection{Experimental Results}

As an illustrative example, we consider a three-dimensional case: $\beta_0 = (2.0, 1.0, 0.5)^\top$ and $\gamma_0 = (1.5, 0.5, 0.0)^\top$.
That is, each dimension of X has a different level of heteroscedastic errors.
In this setting, Figure~\ref{fig:comparison_posterior} presents a side-by-side comparison of the posterior distributions obtained by HDBEN and Bayesian Lasso in a three-dimensional setting. 
While the Bayesian Lasso struggles to construct a concentrated posterior for $\beta$, the posterior of HDBEN is concentrated around the true $\beta_0$, confirming its capability to accurately capture both the mean and variance structures.

\begin{table}[t]
    \centering
    \caption{Experimental results: $s_\beta$ vs. $d$. The evaluation metric is $\|\hat{\beta} - \beta_0\|_2$.}
    \label{tab:s_beta_vs_d}
    \scalebox{0.8}{
    \begin{tabular}{cccccc}
        \toprule
         Models & & $d = 250$ & $d = 500$ & $d = 750$ & $d = 1000$ \\
         \midrule
         \multirow{3}{*}{OLS} & $s_{\bm{\beta}} = 10$ & $14.4303 \pm 5.4454$ & $10.8938 \pm 8.5942$ & $10.0425 \pm 9.5571$ & $8.2810 \pm 2.7911$\\
         & $s_{\bm{\beta}} = 50$ & $23.0768 \pm 10.5738$ & $12.3328 \pm 3.0204$ & $12.1352 \pm 3.7541$ & $11.6468 \pm 2.1492$\\
         & $s_{\bm{\beta}} = 100$ & $20.5320 \pm 9.3309$ & $15.5983 \pm 3.8667$ & $15.0283 \pm 2.6120$ & $16.0777 \pm 2.9098$\\
         \hline
         \multirow{3}{*}{Lasso} & $s_{\bm{\beta}} = 10$ & $7.8348 \pm 4.0640$ & $11.7898 \pm 11.3154$ & $12.4033 \pm 14.0029$ & $11.0569 \pm 5.1340$\\
         & $s_{\bm{\beta}} = 50$ & $15.8072 \pm 8.3154$ & $13.0062 \pm 4.4091$ & $13.5056 \pm 6.2655$ & $13.6658 \pm 4.4933$\\
         & $s_{\bm{\beta}} = 100$ & $15.8406 \pm 6.1762$ & $17.4046 \pm 5.4010$ & $17.2155 \pm 4.4657$ & $19.3454 \pm 5.7117$\\
         \hline
         \multirow{3}{*}{EN} & $s_{\bm{\beta}} = 10$ & $7.7526 \pm 3.3467$ & $10.0699 \pm 8.2446$ & $9.9222 \pm 9.5120$ & $8.6947 \pm 3.2110$\\
         & $s_{\bm{\beta}} = 50$ & $13.8568 \pm 6.2719$& $11.4878 \pm 3.1538$ & $11.7613 \pm 4.1580$ & $11.6748 \pm 2.5598$\\
         & $s_{\bm{\beta}} = 100$ & $13.9064 \pm 4.5270$ & $14.9880 \pm 3.7330$ & $14.8799 \pm 2.7646$ & $16.2865 \pm 3.2823$\\
         \hline
         \multirow{3}{*}{BLasso} & $s_{\bm{\beta}} = 10$ & $14.5684 \pm 5.4992$ & $10.8558 \pm 8.6105$ & $10.0014 \pm 9.6020$ & $8.2486 \pm 2.7998$\\
         & $s_{\bm{\beta}} = 50$ & $23.2775 \pm 10.8511$& $12.1264 \pm 3.1261$ & $12.0137 \pm 3.8322$ & $11.5803 \pm 2.1770$\\
         & $s_{\bm{\beta}} = 100$ & $20.6383 \pm 9.5421$& $15.4003 \pm 3.9388$ & $14.8379 \pm 2.6149$ & $16.0031 \pm 2.9391$\\
         \hline
         \multirow{3}{*}{BEN} & $s_{\bm{\beta}} = 10$ & $14.4293 \pm 5.4445$ & $10.8336 \pm 8.6065$ & $9.9998 \pm 9.5990$ & $8.2477 \pm 2.7992$\\
         & $s_{\bm{\beta}} = 50$ & $23.0471 \pm 10.7106$& $12.1209 \pm 3.1104$ & $12.0156 \pm 3.8289$ & $11.5798 \pm 2.1765$\\
         & $s_{\bm{\beta}} = 100$ & $20.4645 \pm 9.4142$& $15.3912 \pm 3.9325$ & $14.8376 \pm 2.6137$ & $16.0025 \pm 2.9383$\\
         \hline
         \multirow{3}{*}{HDBEN} & $s_{\bm{\beta}} = 10$ & $4.0612 \pm 0.3442$ & $4.4334 \pm 0.4667$ & $4.6319 \pm 0.2634$ & $4.7502 \pm 0.2783$\\
         & $s_{\bm{\beta}} = 50$ & $9.9228 \pm 0.2435$ & $10.4072 \pm 0.4034$ & $10.6357 \pm 0.2464$ & $10.6535 \pm 0.3039$\\
         & $s_{\bm{\beta}} = 100$ & $14.4682 \pm 0.2265$& $14.9926 \pm 0.2745$ & $15.0981 \pm 0.2471$ & $15.1830 \pm 0.2877$\\
         \bottomrule
    \end{tabular}
    }
\end{table}

\begin{figure}
    \centering
    \includegraphics[width=\linewidth]{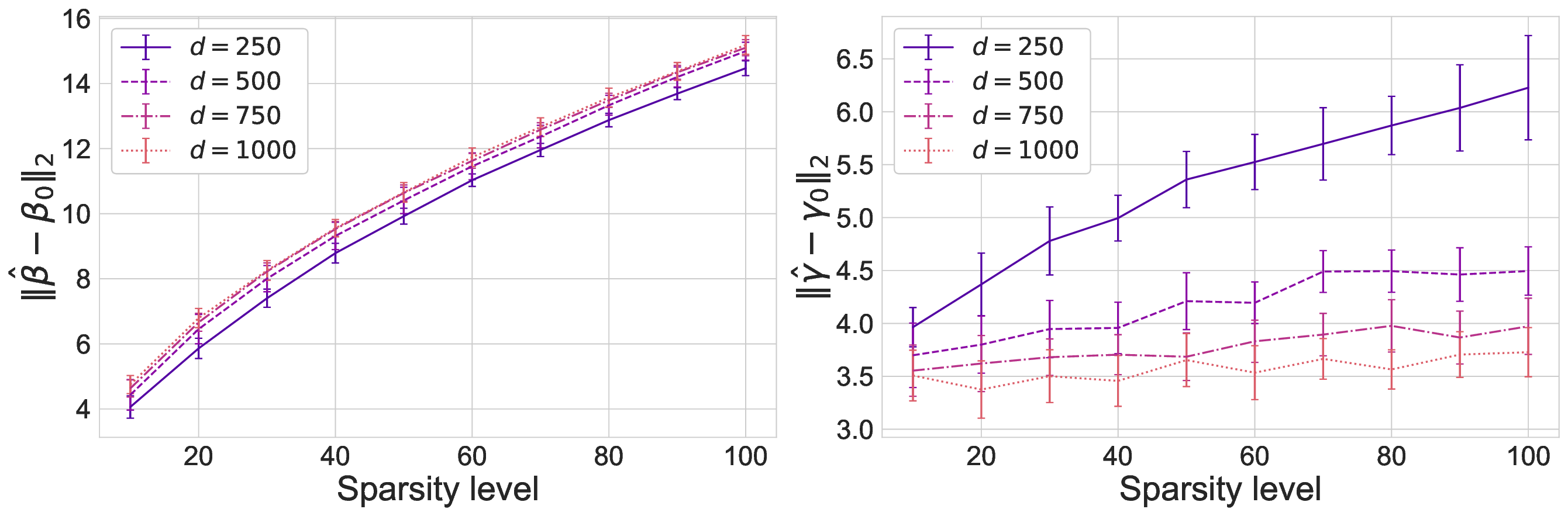}
    \caption{Estimation of $\beta$ and $\gamma$ at different sparsity levels.}
    \label{fig:hdben_beta_gamma_sparsity}
\end{figure}

Table~\ref{tab:s_beta_vs_d} summarizes the $\ell_2$-norm estimation error for the mean parameters $\beta$ across different values of $d$ and sparsity levels.
The results clearly indicate that HDBEN consistently achieves substantially lower estimation errors compared to all competing methods.
In particular, when $s_\beta$ is low, HDBEN reduces the error by nearly a factor of two relative to EN and BLasso, and this advantage persists as the dimensionality increases.
Figure~\ref{fig:hdben_beta_gamma_sparsity} further illustrates that HDBEN accurately recovers the true non-zero coefficients for both $\beta$ and $\gamma$ across different sparsity levels.
The estimated values by HDBEN remain stable, and the posterior distributions are tightly concentrated around the true parameters.

Table~\ref{tab:d_vs_n} reports the performance of all methods as the sample size $n$ varies for fixed $d$.
The proposed method exhibits robust performance, with its estimation error decreasing consistently as $n$ increases.
This behavior is aligned with our theoretical findings on posterior contraction and variable selection consistency.
The performance gap between HDBEN and conventional approaches widens in challenging regimes where the predictor dimension is large relative to the sample size.

In summary, the experimental results demonstrate that HDBEN not only outperforms traditional and Bayesian competitors in terms of estimation accuracy but also reliably recovers the underlying sparsity patterns in both $\beta$ and $\gamma$.
These findings underscore the effectiveness and robustness of HDBEN in handling high-dimensional heteroscedastic regression problems.

\begin{table}[t]
    \centering
    \caption{Experimental results: $d$ vs. $n$. The evaluation metric is $\|\hat{\beta} - \beta_0\|_2$.}
    \label{tab:d_vs_n}
    \scalebox{0.8}{
    \begin{tabular}{cccccc}
        \toprule
         Models & & $n = 50$ & $n = 100$ & $n = 150$ & $n = 200$ \\
         \midrule
         \multirow{3}{*}{OLS} & $d = 100$ &  $9.1556 \pm 6.3961$ & $15.1402 \pm 13.6307$ & $16.3077 \pm 17.4262$ & $16.7525 \pm 18.9128$\\
         & $d = 500$ & $6.0126 \pm 1.9579$ & $5.7663 \pm 1.0106$ & $8.3157 \pm 4.1704$ & $10.8938 \pm 8.5942$\\
         & $d = 1000$ & $5.5414 \pm 1.4911$ & $5.5941 \pm 0.9989$ & $12.2578 \pm 24.8476$ & $8.2810 \pm 2.7911$\\
         \hline
         \multirow{3}{*}{Lasso} & $d = 100$ &  $9.7433 \pm 8.1173$ & $11.3341 \pm 10.3443$ & $12.4596 \pm 16.7532$ & $14.5801 \pm 18.6325$ \\
         & $d = 500$ & $8.7036 \pm 4.9828$ & $6.6206 \pm 2.3459$ & $9.6982 \pm 6.3168$ & $11.7898 \pm 11.3154$\\
         & $d = 1000$ & $8.1921 \pm 4.6877$ & $7.6203 \pm 2.9126$ & $18.3355 \pm 42.7410$ & $11.0569 \pm 5.1340$\\
         \hline
         \multirow{3}{*}{EN} & $d = 100$ & $8.2148 \pm 5.8001$ & $9.6769 \pm 7.0876$ & $11.2582 \pm 13.2780$ & $13.6566 \pm 16.4843$\\
         & $d = 500$ & $6.4681 \pm 2.2510$ & $5.7363 \pm 1.4120$ & $8.0658 \pm 4.3155$ & $10.0699 \pm 8.2446$\\
         & $d = 1000$ & $6.0847 \pm 1.9704$ & $5.9865 \pm 1.5070$ & $12.5819 \pm 24.6938$ & $8.6947 \pm 3.2110$\\
         \hline
         \multirow{3}{*}{BLasso} & $d = 100$ & $9.8897 \pm 7.1867$ & $15.8203 \pm 12.5692$ & $16.0928 \pm 17.1509$ & $16.6701 \pm 18.7033$\\
         & $d = 500$ & $6.1058 \pm 2.0965$ & $5.7071 \pm 1.1432$ & $8.2827 \pm 4.3441$ & $10.8558 \pm 8.6105$\\
         & $d = 1000$ & $5.5642 \pm 1.5254$ & $5.5723 \pm 1.0365$ & $12.2560 \pm 24.9748$ & $8.2486 \pm 2.7998$\\
         \hline
         \multirow{3}{*}{BEN} & $d = 100$ &  $9.2883 \pm 6.5554$ & $11.7706 \pm 7.9308$ & $16.0081 \pm 17.0614$ & $16.6405 \pm 18.6718$\\
         & $d= 500$ & $6.0778 \pm 2.0571$ & $5.6966 \pm 1.1227$ & $8.2458 \pm 4.2916$ & $10.8336 \pm 8.6065$\\
         & $d = 1000$ & $5.5594 \pm 1.5177$ & $5.5702 \pm 1.0344$ & $12.2550 \pm 24.9713$ & $8.2477 \pm 2.7992$\\
         \hline
         \multirow{3}{*}{HDBEN} & $d = 100$ &  $4.3125 \pm 0.2767$ & $3.5513 \pm 0.3632$ & $3.1077 \pm 0.4290$ & $2.9596 \pm 0.3431$ \\
         & $d = 500$ & $4.7140 \pm 0.2598$ & $4.6669 \pm 0.3359$ & $4.5452 \pm 0.2298$ & $4.4334 \pm 0.4667$ \\
         & $d = 1000$ & $4.7987 \pm 0.3441$ & $4.6765 \pm 0.4711$ & $4.7112 \pm 0.2597$ & $4.7502 \pm 0.2783$\\
         \bottomrule
    \end{tabular}
    }
\end{table}

\section{Concluding Remarks}
\label{sec:concluding_remarks}
In this paper, we introduced the Heteroscedastic Double Bayesian Elastic Net (HDBEN), a novel regression framework that simultaneously addresses the challenges of high-dimensionality and heteroscedasticity by jointly modeling the mean and variance functions.
By extending the Bayesian Elastic Net to accommodate a log-linear variance model, HDBEN achieves dual regularization that promotes sparsity and grouping in both the regression coefficients and the variance parameters.
This unified approach not only improves parameter estimation and variable selection but also provides more accurate uncertainty quantification in settings where the error variance varies with the predictors.
Our theoretical developments establish rigorous guarantees for the proposed method.
We showed that under mild regularity conditions, the posterior distribution concentrates around the true parameter values at near-optimal rates and that the method attains variable selection consistency.
Moreover, an asymptotic normality result for the non-zero components further justifies the inferential procedures based on the HDBEN framework in high-dimensional settings.
Extensive simulation studies demonstrate the practical advantages of HDBEN over classical methods.
In various scenarios involving different levels of sparsity, dimensionality, and degrees of heteroscedasticity, HDBEN consistently yields more accurate estimates. The empirical results also highlight its robustness and stability across a wide range of settings.



\clearpage

\bibliographystyle{apa}
\bibliography{main.bib}
\end{document}